\documentclass[runningheads,12pt]{llncs}
\usepackage[margin=3cm]{geometry}
\usepackage{amsmath,amsfonts,amssymb}
\usepackage{mathrsfs}
\usepackage{enumerate}
\usepackage{hyperref}
\newcommand{\GRS}{\mathbf{GRS}}
\newcommand{\RS}{\mathbf{RS}}

\newcommand{\card}[1]{|\left\{#1\right\}|}
\newcommand{\bm}[1]{\boldsymbol{#1}}

\DeclareMathOperator{\tr}{Tr}
\DeclareMathOperator{\Ca}{C_{\bm a}}
\DeclareMathOperator{\supp}{supp}

\DeclareMathOperator{\wt}{wt}

\begin{document}

\title{On Linear Codes with Random Multiplier Vectors and the Maximum Trace Dimension Property\thanks{This research was supported by the Ministry of Culture and Innovation and the National Research, Development and Innovation Office within the Quantum Information National Laboratory of Hungary (Grant No. 2022-2.1.1-NL-2022-00004). Partially founded by NKFIH Grants K129335, K138596, K135885 and FK127906.}}

\author{ 
	M\'arton Erd\'elyi\inst{1}\orcidID{0000-0003-3945-8265} 
	\and P\'al Heged\"us\inst{1}\orcidID{0000-0002-0784-6362} 
	\and S\'andor Z. Kiss\inst{1}\orcidID{0000-0001-9912-8474}
	\and G\'abor P. Nagy\inst{1,2}\orcidID{0000-0002-9558-4197}
}
\institute{Department of Algebra and Geometry,
	Budapest University of Technology and Economics,
	M\H{u}egyetem rkp. 3,
	H-1111 Budapest, Hungary
\and Bolyai Institute,
	University of Szeged,
	Aradi v\'ertan\'uk tere 1,
	H-6720 Szeged, Hungary
\email{\{merdelyi,hegpal,ksandor\}@math.bme.hu, nagyg@math.u-szeged.hu}
}

\maketitle

\begin{abstract}
Let $C$ be a linear code of length $n$ and dimension $k$ over the finite field $\mathbb{F}_{q^m}$. The trace code $\tr(C)$ is a linear code of the same length $n$ over the subfield $\mathbb{F}_q$. The obvious upper bound for the dimension of the trace code over $\mathbb{F}_q$ is $mk$. If equality holds, then we say that $C$ has maximum trace dimension. The problem of finding the true dimension of trace codes and their duals is relevant for the size of the public key of various code-based cryptographic protocols. Let $\Ca$ denote the code obtained from $C$ and a multiplier vector $\bm{a}\in (\mathbb{F}_{q^m})^n$. In this paper, we give a lower bound for the probability that a random multiplier vector produces a code $\Ca$ of maximum trace dimension. We give an interpretation of the bound for the class of algebraic geometry codes in terms of the degree of the defining divisor. The bound explains the experimental fact that random alternant codes have minimal dimension. Our bound holds whenever $n\geq m(k+h)$, where $h\geq 0$ is the Singleton defect of $C$. For the extremal case $n=m(h+k)$, numerical experiments reveal a closed connection between the probability of having maximum trace dimension and the probability that a random matrix has full rank.
\end{abstract}

\section{Introduction}

\subsection*{Code-based post-quantum cryptosystems}
Recent research has focused extensively on quantum computers that use quantum mechanical techniques to solve difficult mathematical computational problems \cite{google2019quantum}. The existence of these potent devices poses a threat to numerous widely used public-key cryptosystems \cite{shor1994polynomial}. 1978 marked the introduction of the first code-based public-key cryptosystem by McEliece \cite{mceliece1978public}. One of the most pressing problems in cryptography today is to reduce the key size and enhance the security level of the McEliece cryptosystem, which is a promising cryptographic scheme for the post-quantum era \cite{Nist}. Error-correcting codes used in code-based cryptographic protocols must be decoded with efficient algorithms. The family of algebraic geometry (AG) codes and their subcodes and subfield subcodes constitute a rich class of such codes. These include the generalized Reed-Solomon, alternant, binary Goppa, and BCH codes. For a survey on decoding AG codes, see \cite{hoholdt1998algebraic}.

The authors of \cite{couvreur2015cryptanalysis,Couvreur2017,couvreur2016polynomial} provided polynomial-time attacks against the McEliece cryptosystem that employs AG codes or their subcodes. In general, evaluation codes do not operate like random codes. This enables a wide variety of attacks against the McEliece cryptosystem based on AG codes. The technique described in \cite{Couvreur2017,couvreur2016polynomial} is inspired by the so-called \emph{filtration attacks} that rely on computing the dimension of the Schur product that makes AG codes distinguishable from random ones. This observation was used by Wieschebrink \cite{wieschebrink2010cryptanalysis} to provide an attack against the McEliece scheme based on subcodes of generalized Reed-Solomon codes \cite{berger2005mask}. Numerous attacks have employed a combination of powerful techniques, such as the filtration method, an error-correcting pair (ECP), or an error-correcting array (ECA), leading to a key recovery attack or a blind reconstruction of a decoding algorithm \cite{couvreur2014distinguisher,Couvreur2017,couvreur2016polynomial}. These vulnerabilities are founded on the \emph{Schur product} operation and a thorough examination of the dimensions of Schur products for specific subcodes.

\subsection*{Key generation of code-based cryptosystems} The key generation process of the code-based scheme starts with a public code $C_0$ and a decoding algorithm $\Delta_0$ which can efficiently correct a certain number of errors. Then, a random seed $\sigma$ and a procedure $\Pi$ are used to construct a code $C$ with a decoding algorithm $\Delta$.

Roughly speaking, the code $C=\Pi(C_0,\sigma)$ represents the public key, while the decoding algorithm $\Delta=\Pi(\Delta_0,\sigma)$ represents the private key. The class of \emph{random alternant codes,} where the starting code $C_0$ is the full support Reed-Solomon code of dimension $k$ over the field of order $q^m$ serves as an illustration. The random seed consists of a pair of vectors of length $n$ over $\mathbb{F}_{q^m}$: the \emph{multiplier} $\bm{a} = (a_1,\ldots,a_n)$, $(a_i\neq 0$), and the \emph{support} $\bm{x} = (x_1,\ldots,x_n)$, ($x_i\neq x_j$ for $i\neq j$). The process $\Pi$ has two main steps: first, compute the generalized Reed-Solomon code $C_1=\GRS_k(\bm{x},\bm{a})$, then compute the subfield subcode $\mathcal{A}_k(\bm{x},\bm{a})=C_1^\perp\cap \mathbb{F}_q^n$ of the dual of $C_1$. Due to Delsarte's theorem, the second step is equivalent to taking the dual of the trace code: $\mathcal{A}_k(\bm{x},\bm{a})=\tr(C_1)^\perp$. (For more precise definitions and references, see Section \ref{sec:prereq}.)

While binary Goppa codes form a subclass of alternant codes, randomness for binary Goppa codes operates distinctly. One starts with the full support Reed-Solomon code $C_0$, where $q=2$, $k=2t$. The seed consists of the support $\bm{x}$, and the monic irreducible polynomial $g(X)$ of degree $t$ over $\mathbb{F}_{q^m}$. The multiplier $\bm{a}$ is defined by $a_i=1/g(x_i)$, and the result is the alternant code $\mathcal{A}_k(\bm{x},\bm{a})$. In both cases, the scheme's cryptographic strength depends on taking the subfield subcode or, equivalently, taking the trace code. Existing known mathematical techniques have yet failed to grasp the essence of these two operations. In particular, it is difficult to determine the true dimension of subfield subcodes and trace codes in general.

\subsection*{Random trace codes and their dimension}
Subfield subcodes and trace codes are linked by duality. This paper deals with the dimension problem of trace codes. Let $q$ be a prime power, $m,k,n$ positive integers.
We extend the trace map $\tr:\mathbb{F}_{q^m} \to \mathbb{F}_q$ to vectors and matrices. For a linear subspace $C\leq \mathbb{F}_{q^m}^n$, we write $\tr(C)=\{\tr(\bm{x}) \mid \bm{x}\in C\}$. For the linear code $C\leq \mathbb{F}_{q^m}^n$ of dimension $k$, we have the obvious upper bound $\dim(\tr(C))\leq mk$, and we say that $C$ has \emph{maximum trace dimension,} if the equality holds. Assume that the $\mathbb{F}_q$-linear code $C=\Pi(C_0,\sigma)$ is constructed using a $\mathbb{F}_{q^m}$-linear code $C_0$ and a random seed $\sigma$. Then we may inquire about the probability
\[\mathrm{Prob}(\text{$C=\Pi(C_0,\sigma)$ has maximum trace dimension}),\]
a value which depends solely on $C_0$. This probability has already been estimated using numerical experimentation for binary Goppa codes of the classic McEliece scheme (see Sections 2.2.2 and 4.2 of \cite{ClassicMcEliece2020}), and for random alternant codes \cite{Mora2023}.

The focus on this probability is mainly theoretical; however, bounds on the proportion of random alternant codes with maximum trace dimension are beneficial in understanding the complexity of the algorithms used in the key generation process of code-based cryptography, as well as the size of public keys. 

In this paper, we prove a lower bound for the probability of maximum trace dimension in the probability model of random multipliers.

\begin{theorem}\label{thm:defect}
Let $C$ be an $[n,k,d]_{q^m}$-code and let $h=n+1-k-d$ be its Singleton defect. Let $P_C$ denote the proportion of multiplier vectors $\bm{a}=(a_1,\ldots,a_n)\in (\mathbb{F}_{q^m}^*)^n$ such that the linear code
\begin{align*} %
\Ca=\{(a_1x_1,\ldots,a_nx_n)\mid \bm{x}\in C\}
\end{align*}
has maximum trace dimension.  Then
\begin{align} \label{eq:rhsmain}
P_C \geq 1-\frac{1-q^{-m(h+k)}}{(q-1)q^{n-m(h+k)}}.
\end{align}
In particular, if $n\geq m(k+h)$, or equivalently $d\geq n(1-1/m)+1$, then $P_C>0$.
\end{theorem}

Our proof uses double counting methods that involve the weight distribution of the dual code $C^\perp$. We apply recent results by Meneghetti, Pellegrini and Sala \cite{Meneghetti2022} that relate the weight distribution to numerical properties of the code that can be computed if the Singleton defect is small. For our purposes, the most important such property is the number of $k\times v$ submatrices of rank $r$ of the generator matrix.

Except for the case $q=2$ and $n=m(h+k)$, Theorem \ref{thm:defect} implies $P_C\geq 1/2$. This means that the Monte Carlo method of generating a random code $\Ca$ of maximum trace dimension is very effective. For $q=2$ and $n=m(h+k)$, further research is necessary.

\subsection*{Maximum trace dimension probabilities of AG codes}
Algebraic geometry (AG) codes are linear error-correcting codes constructed from algebraic curves over finite fields, generalizing the Reed-Solomon code concept. They are defined by \emph{evaluating functions} or by using \emph{residues of differentials}. Their parameters can be derived from well-known theorems in algebraic geometry. Our notation and terminology on algebraic plane curves over finite fields, their function fields, divisors, and Riemann-Roch spaces are conventional; see, for example, \cite{stichtenoth2009algebraic}.

Let $\mathcal{X}$ be a smooth algebraic curve over the finite field $\mathbb{F}_{q^m}$. Let $P_1, P_2, \cdots,P_n$ be pairwise distinct places of $\mathcal{X}$, and $D$ is the divisor $D=P_1+...+P_n$. Let $G$ be another divisor with support disjoint from $D$. The Riemann-Roch theorem enables us to estimate the dimension, the minimum distance and the Singleton defect of AG codes. These, together with Theorem \ref{thm:defect} imply a lower bound for the probability of maximum trace dimension of AG codes.
\begin{theorem} \label{thm:agcodes}
Let $C=C_L(D,G)$ be a functional AG code of length $n=\deg(D)$ over the finite field $\mathbb{F}_{q^m}$, $m>1$. If $\deg(G)\leq n/m-1$, then
\begin{align} \label{eq:rhsag}
P_C \geq 1-\frac{1-q^{-m(\deg(G)+1)}}{(q-1)q^{n-m(\deg(G)+1)}}.
\end{align}
\end{theorem}

\subsection*{Rank properties of random matrices in other probabilistic models}
The rank properties of random matrices over finite fields have been extensively studied as a problem in combinatorial graph theory and other contexts, including coding theory and code-based post-quantum cryptography.  For the probabilistic paradigm, there are a variety of alternatives \cite{Cooper2000,Cooper2000a,salmond2016rank,Studholme2006,Studholme2010}. One possibility is to choose each entry of the matrix independently and uniformly at random from the field. This can be extended to non-uniform distributions, which may or may not depend on the matrix entry's position. Studholme and Blake \cite{Studholme2010} studied \emph{windowed random matrices,} where the nonzero elements of each column are restricted to fall within a window of length $w$, beginning at a randomly chosen row. In \cite{salmond2016rank}, the authors prove that the probability that a random matrix has full rank cannot increase if we fix any number of additional elements to be identically zero.

Let $A$ be an $n\times n$ matrix over the finite field $\mathbb{F}_q$, whose entries are chosen uniformly at random. As $n\to \infty$, the probability that $A$ has rank $n$ converges very fast to the value
\begin{align} %
S(q)=\prod_{i=1}^{\infty}\left(1-\frac{1}{{q^i}}\right).
\end{align}
$S(q)$, which is independent of $n$, is also called the $q$-Pochhamer symbol $(1/q;1/q)_\infty$, see \cite{enwiki:1109461763}. For $q=2$, a good estimate for $S(2)$ is $0.2888$. Let $V$ be an $\mathbb{F}_q$-space of dimension $n$, and take $n$ nonzero vectors uniformly at random from $V\setminus\{0\}$. The probability that the vectors are linearly independent also converges to $S(q)$ very fast if $n\to\infty$.

We performed numerical experiments for Reed-Solomon codes $C=\RS_k(\bm{x})$ over $\mathbb{F}_{q^m}$, where $k,m$ are positive integers, $q=2$ or $q=3$, and $x_1,\ldots,x_{km}$ are random distinct elements of $\mathbb{F}_{q^m}$. Therefore, $C$ has length $n=km$, and Singleton defect $h=0$. We observed that the probability that $C$ has maximum trace dimension is near to the value $S(q)$.

\subsection*{Outline of the article}

Notation and classical prerequisites on linear codes are given in Section \ref{sec:prereq}. Section \ref{sec:maxtdim} collects basic properties and examples of codes which have maximum trace dimension. In Section \ref{sec:randalt} we deal with the dimension problem of random alternant codes. Sections \ref{sec:proof1} and \ref{sec:proof2} contain detailed proofs of the main theorems. The basic concepts of AG codes are also presented in Section \ref{sec:proof2}.

\section{Prerequisites from coding theory} \label{sec:prereq}

Let $q$ be a prime power and let $m,n,k$ be positive integers such that $mk\leq n\leq q^m$. Let $x_1,\ldots,x_n$ be distinct elements of $\mathbb{F}_{q^m}$. The \emph{Reed-Solomon code} $\RS_k(\bm{x})$ is defined by the generator matrix
\begin{align} \label{eq:RSgenmat}
G=\begin{bmatrix}
1&1&\cdots&1\\
x_1&x_2&\cdots&x_n\\
\vdots&\vdots&&\vdots\\
x_1^{k-1}&x_2^{k-1}&\cdots&x_n^{k-1}
\end{bmatrix}.
\end{align}
The vector $\bm{x}$ is called the \emph{support} of the Reed-Solomon code. $\RS_k(\bm{x})$ has dimension $k$ and minimum distance $d=n-k+1$. It is a \emph{maximum distance separable} (MDS) code with Singleton defect $h=0$. Let $a_1,\ldots,a_n$ be nonzero elements of $\mathbb{F}_{q^m}$. The \emph{generalized Reed-Solomon code} $\GRS_k(\bm{x},\bm{a})$ has generator matrix
\[ G'=\begin{bmatrix}
a_1&a_2&\cdots&a_n\\
a_1x_1&a_2x_2&\cdots&a_nx_n\\
\vdots&\vdots&&\vdots\\
a_1x_1^{k-1}&a_2x_2^{k-1}&\cdots&a_nx_n^{k-1}
\end{bmatrix}. \]
Clearly, $\GRS_k(\bm{x},\bm{a})$ and $\RS_k(\bm{x})$ have same parameters. In particular, generalized Reed-Solomon codes are MDS. If $n=q^m$, then $\{x_1,\ldots,x_n\}=\mathbb{F}_{q^m}$ and the codes are said to have \emph{full support.} The dual code of $\GRS_k(\bm{x},\bm{a})$ is again a generalized Reed-Solomon code $\GRS_{n-k}(\bm{x},\bm{b})$, with the same support $\bm{x}$. The Berlekamp-Massey algorithm provides an efficient decoding algorithm for Reed-Solomon codes, which can correct up to $\lfloor \frac{d-1}{2} \rfloor=\lfloor \frac{n-k}{2} \rfloor$ errors. If the multiplier vector is given, then this algorithm can also be used to decode generalized Reed-Solomon codes.

Let $C$ be a linear code of length $n$, dimension $k$ and minimum distance $d$, defined over the finite field $\mathbb{F}_{q^m}$. The \emph{subfield subcode} or \emph{restricted code} of $C$ is
\[ C|_{\mathbb{F}_q} = C\cap \mathbb{F}_q^n.\]
We extend the trace map $\tr:\mathbb{F}_{q^m} \to \mathbb{F}_q$ to vectors and matrices entry-wise. We define the \emph{trace code} of the linear $C\leq \mathbb{F}_{q^m}^n$ by
\[ \tr(C)=\{\tr(\bm{x}) \mid \bm{x}\in C\}. \]
Clearly, $\tr(C)$ is an $\mathbb{F}_q$-linear code of length $n$. Let $\bm{x}_1,\ldots,\bm{x}_k$ be a basis of $C$, and let $\beta_1,\ldots,\beta_m$ be a basis of $\mathbb{F}_{q^m}$ over $\mathbb{F}_q$. Then the vectors $\tr(\beta_i \bm{x}_j)$, ($1\leq i \leq m$, $1\leq j \leq k$) span the trace code $\tr(C)$. This implies the obvious upper bound $\dim(\tr(C))\leq km$ for the dimension of the trace code. We say that $C$ has \emph{maximum trace dimension,} if
\[ \dim_{\mathbb{F}_q}(\tr(C)) = m\dim_{\mathbb{F}_{q^m}} (C).\]
According to Delsarte's theorem \cite{Del75},
\[ (\tr(C))^\perp = (C^\perp)|_{\mathbb{F}_q}, \]
which shows that subfield subcodes and trace codes are basically dual objects. This yields the obvious lower bound
\[ \dim (C|_{\mathbb{F}_q}) \geq n-m(n-k)\]
for the dimension of the subfield subcode. The minimum distance of $C|_{\mathbb{F}_q}$ is at least the minimum distance of $C$. Moreover, subfield subcodes inherit the decoding algorithms of their parent code.

An \emph{alternant code} is defined as the subfield subcode of a generalized Reed-Solomon code
\[ \mathcal{A}_k(\bm{x},\bm{a}) = (\GRS_k(\bm{x},\bm{a})^\perp)|_{\mathbb{F}_q},\]
or equivalently, as the dual code of the trace code of a generalized Reed-Solomon code:
\[ \mathcal{A}_k(\bm{x},\bm{a}) = \tr(\GRS_{k}(\bm{x},\bm{a}))^\perp.\]
The integer $k$ is referred to as the \emph{degree} of the alternant code, and $m$ as its \emph{extension degree.} The vector $\bm{x}$ is the \emph{support,} the vector $\bm{a}$ is the \emph{multiplier} of the alternant code. In the sequel, even without explicitly saying it, we assume that the entries of the support vector are distinct, and the entries of the multiplier vector are different from zero.

The obvious lower bound for the dimension of the alternant code is
\begin{align*}
\dim(\mathcal{A}_k(\bm{x},\bm{a}))\geq n-mk.
\end{align*}
Given the support and the multiplier, the Berlekamp-Massey algorithm can correct up to $\lfloor \frac{k}{2} \rfloor$ errors for the alternant code $\mathcal{A}_k(\bm{x},\bm{a})$.

\section{The maximum trace dimension property} \label{sec:maxtdim}

In this section, we prove a collection of properties of codes having the maximum trace dimension. At the end of the section, we present a class of examples which shows that Theorem \ref{thm:defect} is close to being sharp asymptotically.

In the sequel, $C$ denotes an $\mathbb{F}_{q^m}$-linear code of length $n$, dimension $k$ and minimum distance $d$. 

\begin{definition}
Define the \emph{support} of $C$ as 
\begin{align} %
\supp(C)=\{i\in\{1,2,\dots,n\}\mid\exists \bm{x}\in C: x_i\neq 0\}.
\end{align}
For an integer $i$ define
\begin{align} %
d_i(C)=\min_{\substack{D\leq C\\ \dim(D)=i}}|\supp(D)|.
\end{align}
\end{definition}

Note that $d_1(C)=d$ is the minimum distance. Clearly, $\supp(C)=\supp(\Ca)$ and $d_i(C)=d_i(\Ca)$ for each multiplier vector $\bm{a}$. Furthermore, $\dim(C)\leq |\supp(C)|$. 

The proofs of the following lemmas are straightforward consequences of the definitions. 

\begin{lemma}
The following are equivalent:
\begin{enumerate}[(i)]
\item The code $C$ has maximum trace dimension.
\item All $\mathbb{F}_{q^m}$-linear subspaces of $C$ have maximum trace dimension.
\item For all $\bm{x}\in C\setminus\{\bm{0}\}$, $\tr(\bm{x})\neq \bm{0}$.
\item $C\cap K=\{\bm{0}\}$, where $K$ is the kernel of the trace map $\tr:\mathbb{F}_{q^m}^n\to \mathbb{F}_q^{n}$.
\end{enumerate}
\end{lemma}

\begin{lemma} \label{lm:conj}
Assume that for some multiplier vector $\bm{a}$, $\Ca$ has maximum trace dimension. Then we have $d_i(C)\geq im$ for all $1\leq i\leq k$.
\end{lemma}
\begin{proof}
If $D\leq C$, $\dim(D)=i$ such that $|\supp(D)|<im$, then 
\[supp(\tr(D_{\bm{a}}))\subseteq\supp(D_{\bm{a}})=\supp(D) \]
and
\begin{align*} %
\dim(\tr(D_{\bm{a}})) &\leq |\supp(\tr(D_{\bm{a}}))| = |\supp(\tr(D))| \\
&<im = m\dim(D) =m\dim(D_{\bm{a}}).
\end{align*}
Therefore, $D_{\bm{a}}$ and $\Ca$ have no maximum trace dimension.
\qed\end{proof}
We conjecture that the converse of Lemma \ref{lm:conj} holds as well.

As the following examples show, the proportion of multiplier vectors producing a maximum trace dimension code is related to the probability of a random matrix to have full rank. Let $A$ be an $n\times n$ matrix whose entries are chosen from $\mathbb{F}_q$ uniformly at random. The probability for $A$ to have maximum rank $n$ is
\begin{align} %
S_1(n,q)=\prod_{i=1}^{n}\left(1-\frac{1}{{q^i}}\right).
\end{align}
As $n\to \infty$, $S_1(n,q)$ converges very fast to the value
\begin{align} %
S(q)=\prod_{i=1}^{\infty}\left(1-\frac{1}{{q^i}}\right).
\end{align}
Let $V$ be an $\mathbb{F}_q$-space of dimension $n$, and take $n$ nonzero vectors uniformly at random from $V\setminus\{0\}$. The probability for the vectors to be linearly independent is
\begin{align} %
S_2(n,q)=\prod_{j=0}^{n-1}\frac{q^n-q^j}{q^n-1} = \prod_{i=1}^{n}\left(1-\frac{1}{{q^i}}\right) \left(1+\frac{1}{q^n-1}\right)^n.
\end{align}
As the last factor converges to $1$ very fast, $S_2(n,q)\to S(q)$ very fast, see Figures \ref{fig:qis2} and \ref{fig:qis3}. In fact, if $n>20$, then $S(q)$ is a good practical approximation for $S_1(n,q)$ and $S_2(n,q)$.

\begin{figure}
\setlength{\tabcolsep}{5pt} \centering
\caption{$q=2$} \label{fig:qis2}
\begin{tabular}{llll}
$n$ & $S_1(n,q)$ & $S_2(n,q)$ & $S_1(n,q)-S(q)$\\ \hline
5 & 0.298004150390625 & 0.349271971075915 & 0.00921605530402259 \\
10 & 0.289070298419749 & 0.291908472309700 & 0.000282203333146547 \\
15 & 0.288796908379162 & 0.288929141393520 & 8.81329255975061e-6 \\
20 & 0.288788370496567 & 0.288793878752760 & 2.75409964223261e-7 \\
25 & 0.288788103693158 & 0.288788318857146 & 8.60655607892724e-9 \\
30 & 0.288788095355557 & 0.288788103424204 & 2.68954858384518e-10 \\
35 & 0.288788095095007 & 0.288788095389177 & 8.40483238562229e-12 \\
40 & 0.288788095086865 & 0.288788095097371 & 2.62623256475081e-13 \\
45 & 0.288788095086611 & 0.288788095086980 & 8.16013923099490e-15 \\
50 & 0.288788095086603 & 0.288788095086615 & 2.22044604925031e-16
\end{tabular}
\end{figure}

\begin{figure}
\setlength{\tabcolsep}{5pt} \centering
\caption{$q=3$} \label{fig:qis3}
\begin{tabular}{llll}
$n$ & $S_1(n,q)$ & $S_2(n,q)$ & $S_1(n,q)-S(q)$\\ \hline
5 & 0.561280381843718 & 0.572973321315295 & 0.00115430391576976 \\
10 & 0.560130820850226 & 0.560225688332595 & 4.74292227792272e-6 \\
15 & 0.560126097446024 & 0.560126682988612 & 1.95180757112112e-8 \\
20 & 0.560126078008270 & 0.560126081221122 & 8.03216382294636e-11 \\
25 & 0.560126077928279 & 0.560126077944806 & 3.30735439035834e-13 \\
30 & 0.560126077927950 & 0.560126077928031 & 1.44328993201270e-15 \\
\end{tabular}
\end{figure}

\begin{lemma}
Let $C$ be the $m$-fold repetition code over $\mathbb{F}_{q^m}$. The probability that $\Ca$ has maximum trace dimension for a random multiplier vector $\bm{a}$ is $S_2(m,q)$. In practice, if $m\geq 20$, then $S(q)$ is a good approximation for this probability.
\end{lemma}

Let $C_i$ be linear $[n_i,k_i]_{q^m}$-codes, $i=1,2$. Their sum $C_1+C_2$ is a linear $[n_1+n_2,k_1+k_2]_{q^m}$-code whose codewords are $(\bm{x}_1,\bm{x}_2)$ with $\bm{x}_i\in C_i$. The minimum distance of the sum is $d(C_1+C_2)=\min(d(C_1),d(C_2))$.

\begin{lemma}
Let $C,C'$ be $\mathbb{F}_{q^m}$-linear codes, and $D=C+C'$ their sum. Then $P_D=P_CP_{C'}$, where $P_C$ is as defined in Theorem \ref{thm:defect}.
\end{lemma}

Let $C$ be the $k$-fold sum of the $m$-fold repetition code. Clearly, $C$ has length $n=mk$, dimension $k$ and minimum distance $d=m$. The last two lemmas imply that the proportion $P_C$ of multiplier vectors with maximum trace dimension is approximately $P_C \approx S(q)^k$, which tends to zero if $k\to \infty$. In particular, we cannot expect $P_C$ to be close to $1$ just because $k$ and $m$ are large. However, $P_C>0$, so there is a multiplier $\bm a$ such that $C_{\bm a}$ has the maximum trace dimension. On the other hand, $d_i(C)=im$ for all $1\leq i\leq k$, showing that Lemma \ref{lm:conj} is sharp.

Clearly, if $n<mk$ then $P_C=0$. In Theorem \ref{thm:defect} we see that $n\geq m(h+k)$ implies $P_C>0$. The question whether $P_C$ is zero or not is open for the interval $[mk,m(h+k)-1]$. The following class of examples has Singleton defect $h\approx \log_q(k)$, hence the interval is small. Still, the condition $n=mk$ is not enough to ensure $P_C>0$. In other words, Theorem \ref{thm:defect} is close to being sharp asymptotically.

\begin{proposition}
For all prime power $q$ and integers $m>2$, $2\leq k \leq q^m/m$ there exists an $\mathbb{F}_{q^m}$-linear code $C'=C'(q,m,k)$ of length $n=mk$, dimension $k$ and Singleton bound $h=m$, such that $P_{C'}=0$.
\end{proposition}
\begin{proof}
Let $n'=m(k-1)-1$ and let $x_1,\ldots,x_n$ be distinct elements of $\mathbb{F}_{q^m}$ such that $x_{n'+1},\ldots,x_n\neq 0$. Let $C'=C'(q,m,k)$ be the code with generator matrix
\begin{align*}
G'=\begin{bmatrix}
1&1&\cdots&1&0&\cdots&0\\
x_1&x_2&\cdots&x_{n'}&0&\cdots&0\\
\vdots&\vdots&&\vdots&\vdots&&\vdots\\
x_1^{k-2}&x_2^{k-2}&\cdots&x_{n'}^{k-2}&0&\cdots&0 \\
x_1^{k-1}&x_2^{k-1}&\cdots&x_{n'}^{k-1}&x_{n'+1}^{k-1}&\cdots&x_{n}^{k-1}
\end{bmatrix}.
\end{align*}
Let $D'$ be the subcode generated by the first $k-1$ rows. As $k-1\leq n'$, we have $\dim(D')=k-1$. Moreover, $D'$ has support $\{1,2,\dots,n'\}$, hence $|\supp(D')|=n'=m(k-1)-1<m\dim(D')$. Lemma \ref{lm:conj} implies that $P_{C'}=0$.

Now we compute the minimum distance of $C'$. Take any linear combination $\bm{c}$ of the rows of $G'$. Write $\bm{c}'=(c_1,c_2,\dots,c_{n'})$ and $\bm{x}'=(x_1,x_2,\dots,x_{n'})$. If the last row has zero coefficient, then the last $m+1$ coordinates are $0$ and $\bm{c}'\in\RS_{k-1}(\bm{x}')$. So
\[ \wt(\bm{c})\geq n'-(k-1)+1=(m-1)(k-1), \]
and equality occurs for some $\bm{c}$. If the last row has nonzero coefficient, then the last $m+1$ coordinates of $\bm{c}$ are nonzero and $\bm{c}'\in\RS_k(\bm{x}')$. So
\[ \wt(\bm{c})\geq n'-k+1+m+1>(m-1)(k-1). \]
Thus the minimal distance of $C'(m,k)$ is indeed $d=(m-1)(k-1)$ and $h=m$.
\qed\end{proof}

\section{The dimension of random alternant codes} \label{sec:randalt}

In numerical experiments, one observes that the dimension of random alternant codes typically attains the obvious lower bound; see \cite{Mora2023}. In this short section we derive a proof for this observation from Theorem \ref{thm:defect}. We show that if the length of the random alternant code exceeds $mk$, then the dimension is $n-mk$ with high probability. In particular, this is the case for most random alternant codes of full support.

\begin{definition}
Given the field of definition $\mathbb{F}_q$, the degree $k$ and the extension degree $m$, the random alternant code is a code $\mathcal{A}_k(\bm{x},\bm{a})$, where the support $\bm{x}$ and the multiplier $\bm{a}$ are chosen uniformly at random.
\end{definition}

\begin{proposition}
Let $q$ be a prime power and $m,n,k$ be positive integers such that $mk\leq n \leq q^m$. The random alternant code of length $n$, degree $k$, extension degree $m$ over $\mathbb{F}_q$ has dimension $n-mk$ with probability at least
\begin{align*} %
1-\frac{1-q^{-mk}}{(q-1)q^{n-mk}}.
\end{align*}
\end{proposition}
\begin{proof}
The dual of the alternant code is $\tr(\GRS_{k}(\bm{x},\bm{a}))$. Since $\GRS_{k}(\bm{x},\bm{a})$ is MDS of dimension $k$, Theorem \ref{thm:agcodes} implies the proposition.
\qed\end{proof}

\section{Proof of Theorem \ref{thm:defect}} \label{sec:proof1}

In this section, we use the notation of Theorem \ref{thm:defect}. We describe the average cardinality of $\tr(\Ca)^\perp$, $\bm{a} \in (\mathbb{F}_{q^m}^*)^n$, with the help of the weight distribution of the dual code $C^\perp$. Let us introduce the following notation:
\begin{definition}
Let $\wt:C\to\mathbb N$ denote the Hamming weight and
\begin{align} %
B_w=\card{\bm c\in C^\perp\mid\wt(c)=w}
\end{align}
for $0\leq w\leq n$ the weight distribution of $C$. Then let
\begin{align} %
\lambda(C)=\sum_{w=0}^nB_w\left(\frac{q-1}{q^m-1}\right)^w.
\end{align}
For $0\leq r\leq v\leq n$ let
\begin{align} \label{eq:NG}
N_G(v,r)=\card{k\times v \text{ submatrices of } G \text{ with rank } r},
\end{align}
where $G\in \mathbb F_{q^m}^{k\times n}$ is a generator matrix of $C$.
\end{definition}

\begin{proposition}\label{prop:lambdacavgdim}
	We have the following average form:
\begin{equation}\label{eq:lambda_average}
\lambda(C)=\frac{1}{\mid \left(\mathbb F_{q^m}^*\right)^n \mid }\sum_{\bm a\in\left(\mathbb F_{q^m}^*\right)^n}q^{n-\dim(\tr(\Ca))}.
\end{equation}
\end{proposition}
\begin{proof}
For $\bm{a}\in (\mathbb{F}_{q^m}^*)^n$, we write $\bm{a}^{-1}=(a_j^{-1})_{1\leq j\leq n}$. We double-count the set
\begin{align} %
H=\{(\bm{a},\bm{c}) \mid \bm{a}^{-1}\star \bm{c}\in \mathbb{F}_q^n, \bm{a}\in (\mathbb{F}_{q^m}^*)^n, \bm{c}\in C^\perp\}.
\end{align}
For any fixed $\bm{a}$, $(\bm{a},\bm{c})\in H$ if and only if $\bm{a}^{-1}\star\bm{c}\in (C^{\perp})_{\bm a^{-1}}\cap\mathbb{F}_q^n$. By Delsarte's theorem \cite[Theorem 2]{Del75}, we have
\begin{align} %
(C^{\perp})_{\bm a^{-1}}\cap\mathbb{F}_q^n=(\Ca)^\perp\cap\mathbb{F}_q^n=(\tr(\Ca))^\perp.
\end{align}
Hence, $|(C^{\perp})_{\bm a^{-1}}\cap\mathbb{F}_q^n|=q^{n-\dim(\tr(C_{\bm{a}}))}$. This proves
\begin{align} %
|H|=\sum_{\bm a\in\left(\mathbb F_{q^m}^*\right)^n}q^{n-\dim(\tr(\Ca))}.
\end{align}
Let us now fix $\bm c\in C^\perp$. For each $j$, we have
\begin{align} %
\{a_j\in\mathbb{F}_{q^m}^*\mid a_j^{-1}c_j\in\mathbb{F}_q\}=
\begin{cases}
\mathbb{F}_{q^m}^*, & \text{if }c_j=0;\\
c_j\mathbb{F}_q^*, & \text{if }c_j\neq 0.
\end{cases}
\end{align}
Thus
\begin{align} %
\card{\bm a \in (\mathbb{F}_{q^m}^*)^n \mid \bm a^{-1}\star \bm c\in\mathbb{F}_q^n}=(q-1)^{\wt(\bm c)}(q^m-1)^{n-\wt(\bm c)},
\end{align}
summing over all $\bm c\in C^\perp$, we get $|H|= (q^m-1)^n \lambda(C)$.
\qed\end{proof}

As $\dim(\tr(\Ca))\leq km$, each summand on the right-hand side of \eqref{eq:lambda_average} is at least $q^{n-km}$. This gives a lower bound
\begin{equation}
\lambda(C)\geq q^{n-km}.
\end{equation}
The upper bounds of $\lambda(C)$ can be used to find lower bounds on the proportion $P_C$ of multiplier vectors which produce maximum trace dimension codes.
\begin{proposition} \label{prop:Plowerbound}
Assume $\lambda(C)\leq q^{n-km}+D$. Then
\begin{align} %
P_C\geq 1-\frac{D}{(q-1)q^{n-km}}.
\end{align}
\end{proposition}
\begin{proof}
If $\Ca$ does not have maximum trace dimension, then the corresponding summand in \eqref{eq:lambda_average} is at least $q^{n-km+1}$. Therefore
$P_Cq^{n-km}+(1-P_C)q^{n-km+1}\leq\lambda(C)$. The claim follows from a straightforward computation.
\qed\end{proof}

\begin{proposition}
\begin{align} %
\lambda(C)=\left(\frac{q^m-q}{q^m-1}\right)^n\sum_{v=0}^n\left(\frac{q-1}{q^m-q}\right)^v\sum_{r=0}^vN_G(v,r)q^{m(v-r)},
\end{align}
where $N_G(v,r)$ is as defined in (\ref{eq:NG}).
\end{proposition}

\begin{proof}

Applying Proposition 3 of \cite{Meneghetti2022} for $C^\perp$ over $\mathbb{F}_{q^m}$, we get
\begin{align} %
\sum_{s=0}^v\binom{n-s}{v-s}B_{s} = \sum_{r=0}^vN_{G}(v,r)q^{m(v-r)}.
\end{align}
Multiplying with $x^v$ and summing over $0\leq v\leq n$, we obtain
\begin{equation}\label{eq:BsNg}
\sum_{v = 0}^{n}x^v\sum_{s=0}^v\binom{n-s}{v-s}B_{s} = \sum_{v = 0}^{n}x^v\sum_{r=0}^vN_G(v,r)q^{m(v-r)}.
\end{equation}
Changing the order of the summation and using the binomial theorem, the left hand side is
\begin{align} %
\sum_{s = 0}^{n}B_{s}x^{s}\sum_{v=s}^{n}\binom{n-s}{v-s}x^{v-s} = \sum_{s = 0}^{n}B_{s}x^{s}(1+x)^{n-s}.
\end{align}
Let us put $x = \frac{q-1}{q^{m}-q}$, thus $1 + x = \frac{q^{m}-1}{q^{m}-q}$. Then by the definition,
\begin{align} %
\lambda(C) = \sum_{s=0}^nB_s\frac{(q-1)^{s}(q^m-1)^{n-s}}{(q^m-1)^{n}} = \left(\frac{q^{m}-q}{q^{m}-1}\right)^{n}\sum_{s = 0}^{n}B_{s}x^{s}(1+x)^{n-s}.
\end{align}
By (\ref{eq:BsNg}), we have
\begin{align} %
\lambda(C)=\left(\frac{q^{m}-q}{q^{m}-1}\right)^{n}
\sum_{v = 0}^{n}x^v\sum_{r=0}^vN_{G}(v,r)q^{m(v-r)},
\end{align}
hence the proposition.
\qed\end{proof}

\begin{proof}[Proof of Theorem \ref{thm:defect}]
Applying Lemma 4 in \cite{Meneghetti2022} for $C^\perp$, all $k\times(k+h)$ submatrix of $G$ has rank $k$. It follows that the rank of all $k\times v$ submatrix equals $k$ if $v\geq k+h$ and is at least $v-h$ if $v<k+h$.

By using this observation, we can bound the inner sum on the right hand side of the previous Proposition:

\begin{itemize}
    \item For $v\geq k+h$, we have $N_G(v,r)=0$ for $r<k$ and
    \begin{align} \label{eq:KS28}
\sum_{r=0}^vN_G(v,r)q^{m(v-r)}=\binom nvq^{m(v-k)},
\end{align}
    \item for $v<k+h$, we have $N_{G}(v,r) = 0$ for $r<v-h$ and
    \begin{align} \label{eq:KS29}
\sum_{r=0}^vN_G(v,r)q^{m(v-r)}\leq \binom nvq^{mh}.
\end{align}
\end{itemize}

In view of \eqref{eq:KS28}, \eqref{eq:KS29} and $x=\frac{q-1}{q^m-q}$, we obtain
\begin{eqnarray*}
\lambda(C)\le \left(\frac{q^{m}-q}{q^m-1}\right)^{n}\left(\sum_{v = 0}^{k+h-1}\binom{n}{v}x^{v}q^{mh} + \sum_{v = k+h}^{n}\binom{n}{v}x^{v}q^{m(v-k)}\right)\\
= \frac{1}{q^{mk}}\left(\frac{q^{m}-q}{q^m-1}\right)^{n}\left(\sum_{v = 0}^{n}\binom{n}{v}(xq^{m})^{v} + \sum_{v=0}^{k+h-1}\binom{n}{v}x^{v}(q^{m(h+k)}-q^{mv})\right)
\\
\le \frac1{q^{mk}}\left(\frac{q^{m}-q}{q^m-1}\right)^{n}\left((1+xq^m)^n + (q^{m(h+k)}-1)\sum_{v = 0}^{k+h}\binom{n}{v}x^v\right)\\
\leq
\frac{1}{q^{mk}}\left(\frac{q^{m}-q}{q^m-1}\right)^{n}\left((1+xq^m)^n + (q^{m(h+k)}-1)(1+x)^n\right)
\end{eqnarray*}

As $1+x=\frac{q^m-1}{q^m-q}$ and $1+xq^m=q\cdot\frac{q^m-1}{q^m-q}$, we get
\begin{align} %
\lambda(C)\le q^{n-mk} + (q^{mh}-q^{-mk}).
\end{align}

By using Proposition \ref{prop:Plowerbound} we get the lower bound on $P_C$ as in the statement.
\qed\end{proof}

\section{Proof of Theorem \ref{thm:agcodes}} \label{sec:proof2}

Algebraic geometry (AG) codes are linear error-correcting codes constructed from algebraic curves over finite fields, generalizing the Reed-Solomon code concept. They are defined by \emph{evaluating functions} or by using \emph{residues of differentials}. Their parameters can be derived from well-known algebraic geometry theorems. Our notation and terminology on algebraic plane curves over finite fields, their function fields, divisors, and Riemann-Roch spaces are conventional; see, for example, \cite{stichtenoth2009algebraic}.

Let $\mathcal{X}$ be an algebraic curve, that is, an affine or projective variety of dimension one, which is absolutely irreducible and non-singular and whose defining equations are (homogeneous) polynomials with coefficients in $\mathbb{F}_q$. Let $g=g(\mathcal{X})$ be the genus of $\mathcal{X}$. $\mathbb{F}_q(\mathcal{X})$ denotes the function field of $\mathcal{X}$. A \emph{divisor} $D$ of $\mathcal{X}$ is a formal sum $D=n_1P_1+\cdots+n_kP_k$, where $n_1,\ldots,n_k \in \mathbb{Z}$ and $P_1,\ldots,P_k$ are places of $\mathbb{F}_q(\mathcal{X})$. If $n_1,\ldots,n_k\geq 0$, then $D\succeq 0$. If $D,E$ are two divisors and $D-E\succeq 0$, then $D\succeq E$. In the case of a nonzero function $f$ of the function field $\mathbb{F}_{q}(\mathscr{X})$, and a place $P$, $v_P(f)$ stands for the order of $f$ at $P$. If $v_P(f)>0$ then $P$ is a zero of $f$, while if $v_P(f)<0$, then $P$ is a pole of $f$ with multiplicity $-v_P(f)$. The \emph{principal divisor} of a nonzero function $f$ is $\mathrm{Div}(f) = \sum_P v_P(f)P$.

For a divisor $D$, the associated Riemann-Roch space $\mathscr{L}(D)$ is the vector space
\[\mathscr{L}(D)=\{f\in \mathbb{F}_{q}(\mathcal{X})\setminus \{0\} \mid \mathrm{Div}(f)\succeq -D \}\cup \{0\}.\]
The dimension $\ell(D)$ of $\mathscr{L}(D)$ is given by the Riemann-Roch Theorem \cite[Theorem 1.1.15]{stichtenoth2009algebraic}:
\[\ell(D)= \ell(W-D) + \deg D - g + 1,\]
where $W$ is a canonical divisor. We denote the set of \emph{differentials} on $\mathcal{X}$ by $\Omega$. The \emph{differential space} of the divisor $D$ is
\[\Omega(D)=\{dh\in \Omega \mid \mathrm{Div}(dh)\succeq A \}\cup \{0\}.\]
In the following, $P_1, P_2, \cdots,P_n$ are pairwise distinct places on $\mathcal{X}$, and $D$ is the divisor $D=P_1+...+P_n$. Let $G$ be another divisor with support disjoint from $D$. We define two types of AG codes, the \emph{functional} and the \emph{differential codes,} respectively:
\begin{align*}
C_L(D,G)&=\left\lbrace \left( f(P_1),\cdots,f(P_n) \right) \mid f\in \mathscr{L}(G)   \right\rbrace,\\
C_\Omega(D,G)&= \left\lbrace \left( \mathrm{res}_{P_1}(\omega),\cdots,\mathrm{res}
_{P_n}(\omega)\right) \mid \omega \in \Omega(G-D)\right\rbrace.
\end{align*}
These codes are dual to each other, and $C_\Omega(D,G)=C_L(D,K+D-G)$ for a well-chosen canonical divisor $K$. The Riemann-Roch theorem enables us to estimate the dimension and the minimum distance of AG codes:
\[\dim(C_L(D,G)) \begin{cases}
\geq \deg(G)-g+1 & 0\leq \deg(G) \leq 2g-2, \\
= \deg(G)-g+1 & 2g-2 \leq \deg(G) \leq n, \\
\leq \deg(G)-g+1 & n \leq \deg(G) \leq n+2g-2.
\end{cases}\]
The minimum distance of a functional code is at least its \emph{designed minimum distance}
\[\delta_L = n-\deg(G).\]

\begin{proof}[Proof of Theorem \ref{thm:agcodes}]
Let $k$ be the dimension, and $h$ be the Singleton defect of the AG code $C=C_L(D,G)$. Then $h+k=n+1-d\leq n+1-\delta_L=\deg(G)+1$. As the right hand side of \eqref{eq:rhsmain} is monotone decreasing in $h+k$, the formula \eqref{eq:rhsag} follows.
\qed\end{proof}

\section{Conclusion}

We gave a lower bound for the probability that the dimension of the trace code of a linear code with a random multiplier vector attains the obvious upper bound. Our formula only uses the size of the underlying field, the degree of the field extension, and the three main parameters of the code: length, dimension and minimum distance. The result provided a compact formula for the probability that an AG code has maximum trace dimension. We also proved by mathematical means that full support random alternant codes have dimension $n-mk$ with high probability. These pieces of information are useful to understand better the size of public keys in code-based cryptography, and the complexity of Monte Carlo algorithms in the key generation process.

Our approach works for the probabilistic model of random multiplier vectors. Random Goppa codes have a different probability paradigm. Therefore, our results do not solve the dimension problem for random Goppa codes. This needs further research, but we are optimistic that our method can be extended.

Another open case is when the length $n$ of the code is less than the product of the extension degree $m$ and the dimension $k$. Then, one asks for the probability that random multiplier vectors produce codes of dimension $n$. Understanding this could help to extend our results for the parameters $q=2$ and $n=mk$.

\bibliography{max_trace_dim,wpref}
\bibliographystyle{splncs04}

\end{document}